\documentclass[conference]{IEEEtran}
\IEEEoverridecommandlockouts

\usepackage{amsmath,amsthm,amssymb}
\usepackage{graphicx}
\usepackage{dsfont}
\usepackage{framed}

\newcommand{\supp}{\mathsf{supp}}
\newtheorem{theorem}{Theorem}
\newtheorem{lemma}[theorem]{Lemma}

\theoremstyle{definition}
\newtheorem{defn}[theorem]{Definition}

\newcommand{\zo}{\{0,1\}}
\newcommand{\eps}{\epsilon}
\newcommand{\poly}{\mathsf{poly}}
\newcommand{\qpoly}{\mathsf{quasipoly}}
\newcommand{\wgt}{\mathsf{wgt}}

\begin{document}

\setlength{\pdfpageheight}{\paperheight}
\setlength{\pdfpagewidth}{\paperwidth}

\title{Derandomization and Group Testing} \author{ \authorblockN{Mahdi
    Cheraghchi\thanks{ Research supported by the ERC Advanced
      investigator grant 228021 of A.~Shokrollahi.}
    \authorblockA{School of Computer and Communication Sciences\\
      EPFL, Lausanne, Switzerland \\
      mahdi.cheraghchi@epfl.ch }}}

\maketitle

\begin{abstract}

  The rapid development of derandomization theory, which is a
  fundamental area in theoretical computer science, has recently led
  to many surprising applications outside its initial intention. We
  will review some recent such developments related to combinatorial
  group testing. In its most basic setting, the aim of group testing
  is to identify a set of ``positive'' individuals in a population of
  items by taking groups of items and asking whether there is a
  positive in each group.

  In particular, we will discuss explicit constructions of optimal or
  nearly-optimal group testing schemes using ``randomness-conducting''
  functions. Among such developments are constructions of
  error-correcting group testing schemes using randomness extractors
  and condensers, as well as threshold group testing schemes from
  lossless condensers.

\end{abstract}

\section{Introduction}

Combinatorial group testing is a classical problem that dates back to
several decades ago \cite{ref:Dor43} and has recently attracted
increased attention mainly due to its numerous applications in various
theoretical and practical areas. Intuitively, the problem can be
described as follows: Suppose that, in a population of $n$
individuals, it is suspected that up to $d$ of them (known as
\emph{defectives}) carry a certain disease that can be diagnosed by
testing blood samples. Typically, the parameters $d$ is considered to
be substantially smaller than $n$. An economical way of testing the
samples is to pool them in groups. For each pool, one can apply the
test on the combination.  A negative outcome would imply that none of
the samples participating in the pool are infected, whereas a positive
outcome means that at least one of the individuals corresponding to
the group is infected. The challenge is then, to design a pooling
strategy that minimizes the number of tests that have to be performed
in order to identify the exact set of defectives.

Over the decades, numerous constructions and variations of group
testing schemes have been proposed in the literature
(cf.~\cite{ref:groupTesting,ref:DH06} for a review of the major
developments). Among those, \emph{non-adaptive} testing schemes in
which the tests are designed and fixed before any measurements are
performed are of particular interest, especially for applications in
biology. Designing the tests for non-adaptive schemes is an
interesting, and challenging, combinatorial problem that has been
extensively studied in the literature. Straightforward techniques from
the probabilistic method can be used to show that random designs are
able to distinguish the set of defectives using a nearly optimal
number of tests and with overwhelming probability. It is however much
more challenging to \emph{derandomize} this task and come up with an
\emph{explicit}; i.e., deterministic, way of designing the tests that
achieve, or approach, the same qualities as offered by randomized
constructions.  Explicit constructions are also important from a
practical point of view, where a design failure can be costly and has
to be avoided.

Derandomization theory is an area at the core of theoretical computer
science that aims for a systematic study of tools and techniques that
can be used to reduce, or eliminate, the need for randomness in
computational tasks. Some major examples include simulating randomized
algorithms with deterministic ones, and derandomizing probabilistic
combinatorial constructions (cf.~\cite{lucaCombinatorial}).  In
particular, tools from derandomization theory has been recently used
for designing optimal, or nearly optimal, explicit combinatorial group
testing schemes. In this paper we give a simplified exposition of
certain such developments \cite{ref:Che09, ref:Che10}. In particular,
we study:

\begin{enumerate}
\item Highly noise-resilient group testing schemes that reliably
  approximate the set of defectives using a substantially smaller
  number of tests than what required for their exact identification.

\item Explicit group testing schemes for the \emph{threshold model},
  where a test outputs positive if the number of positives present in
  the pool exceeds a certain, arbitrary, threshold.
\end{enumerate}

The main combinatorial tools used for the above-mentioned
constructions are the notions of \emph{lossless expanders},
\emph{randomness extractors} and \emph{condensers} that are major
topics of interest in derandomization theory. For this exposition, we
will only highlight the main ideas and, for that matter, present
proofs mainly for certain restricted cases.

The paper is organized as follows. In Section~\ref{sec:testing} we
review the classical group testing model and some of its variations,
including noisy and threshold models.  In Section~\ref{sec:disjunct}
we revisit a classical known combinatorial property, called
\emph{disjunctness} and see how it is related to \emph{graph
  expansion}. We proceed to introduce constructions of noise-resilient
group testing schemes using expander and extractor graphs in
Section~\ref{sec:noise}. Finally, Section~\ref{sec:threshold}
discusses the more general threshold model and introduces a
construction of non-adaptive schemes for this model using lossless
expander graphs.

\section{Group testing and variations}
\label{sec:testing}

In classical group testing, in formal terms, we wish to identify an
unknown $d$-sparse binary vector; i.e., $x=(x_1, \ldots, x_n) \in
\zo^n$ such that
\[
|\{i\colon x_i = 1\}| \leq d,
\]
by performing a number of \emph{measurements}.  Define the
\emph{support} of $x$ (in symbols, $\supp(x)$), as the set of nonzero
entries of $x$ (known as \emph{positives}).  Each measurement is
specified by a subset of coordinate positions \[S \subseteq [n] :=
\{1, \ldots, n\}\] and outputs a binary value which is positive if and
only if $S$ contains one or more positives; i.e., if $S \cap \supp(x)
\neq \emptyset$.  The main challenge is then to design a measurement
scheme, with a reasonably small number of measurements, so that every
$d$-sparse vector can be uniquely identified from the measurement
outcomes.

In this paper, we are interested in \emph{non-adaptive}
measurements. This is when the set of the coordinate positions
defining each measurement is a priori fixed and does not depend on the
outcome of the previous measurements.  We find it convenient to think
of a non-adaptive measurement scheme as a bipartite graph $G(L,R,E)$,
called the \emph{measurement graph}.  The set of left vertices of this
graph is $L := [n]$, in one-to-one correspondence with the coordinate
positions of $x$, and right vertices (the set $R$) correspond to the
measurements. Naturally, the $i$th right vertex is connected to the
set of coordinate position specified by the $i$th measurement, and
this defines the edge set $E$.

Numerous variations and extensions of classical group testing have
been studied in the literature. In this paper, we mention the
following:

\begin{itemize}
\item \emph{Noisy group testing:} In this variation, measurement
  outcomes may be incorrect. In particular, we may allow \emph{false
    positives} (i.e., when a negative outcome is read positive),
  \emph{false negatives} (when a positive becomes negative), or
  both. The nature of bit flips might be \emph{stochastic}, i.e., an
  outcome flips with a certain probability, or \emph{adversarial},
  i.e., an adversary may arbitrarily flip the measurement outcomes
  while being limited only in the number of bit flips.

\item \emph{Threshold model:} This model was introduced by Damaschke
  \cite{ref:thresh1} as a natural extension of classical group
  testing.  The difference between classical group testing and
  threshold testing is that, in the threshold model, a measurement
  specified by a set $S$ of the coordinate positions outputs positive
  if and only if
  \[
  |S \cap \supp(x)| \geq u,
  \]
  i.e., when there are at least $u$ positives\footnote{ In the
    threshold model, we implicitly assume that the support size of the
    unknown vector $x$ is guaranteed to be at least $u$, since
    otherwise all measurement outcomes would have to be negative.  }
  in the pool $S$, for some fixed constant parameter $u$. We will use
  the shorthand $u$-threshold testing for this model. Obviously,
  classical group testing corresponds to the special case
  $u=1$. Damaschke also considers a \emph{positive-gap} threshold
  model that is characterized by lower and upper thresholds $\ell, u$
  where $\ell < u$.  In this model, a measurement outputs positive if
  there are $u$ or more positives in the pool, negative if there are
  no more than $\ell$ positives, and may behave arbitrarily
  otherwise. The gap parameter is defined as $g := u - \ell -1$. Thus,
  $u$-threshold testing is the special case when $g=0$. For the sake
  of this exposition, we only focus on this \emph{gap-free} case, but
  point out that our discussions extend to the positive-gap case in a
  straightforward manner.
\end{itemize}

For a measurement graph $G$ and sparse vector $x \in \zo^n$, we will
use the notation $G[x]$ for the binary vector of measurement outcomes
resulting from the measurements specific by the graph $G$, and more
generally, $G_u[x]$ for the vector of measurement outcomes in the
$u$-threshold model.

\section{Disjunctness and Expansion}
\label{sec:disjunct}

The graph-theoretic property required for the classical group testing
model is the following \emph{disjunctness} property. We will use the
notation $\Gamma(v)$ for the set of neighbors of a vertex $v$ in a
graph and $\Gamma(S)$ for the set of neighbors of a subset $S$ of
vertices, i.e.,
\[
\Gamma(S) := \cup_{v \in S} \Gamma(v).
\]

\begin{defn} \label{def:disjunct} A bipartite graph $G(L,R,E)$ is
  called \emph{$(d,e)$-disjunct} if, for every left vertex $i \in L$
  and every set $S \subseteq L$ such that $|S| \leq d$ and $i \notin
  S$, we have
  \[
  |\Gamma(i) \setminus \Gamma(S)| > e.
  \]
  We refer to the elements of $\Gamma(i) \setminus \Gamma(S)$ as
  \emph{distinguishing vertices}. The parameter $e$ is called the
  \emph{noise tolerance} and a $(d,0)$-disjunct matrix is simply
  called $d$-disjunct.
\end{defn}

By rephrasing the standard results in classical group testing, we see
that disjunctness is the key combinatorial property needed for group
testing. In particular, $d$-disjunct measurement graphs can uniquely
identify $d$-sparse vectors, as stated below.

\begin{lemma} \label{lem:disjunct} Let $G$ be a $(d,e)$-disjunct
  graph. Then for every distinct pairs $x, x' \in \zo^n$ of $d$-sparse
  vectors, we have $\Delta(G[x], G[x']) > e$, where $\Delta(\cdot)$
  denotes the Hamming distance between vectors.
\end{lemma}

\begin{proof}
  Without loss of generality, take any $i \in \supp(x) \setminus
  \supp(x')$ and $S := \supp(x')$. By the disjunctness property, we
  have that $D := \Gamma(i) \setminus \Gamma(S)$ has more than $e$
  distinguishing vertices in it. Now one can immediately see that
  $G[x]$ is positive at all positions corresponding to $D$ (since they
  are connected to $i$) while $G[x']$ is negative (since they are not
  connected to any position on the support of $x'$).
\end{proof}

The noise tolerance $e$ is called so since, according to the above
lemma, a larger $e$ would make the measurements outcomes further
apart, allowing for more resilience against measurement errors in
noisy group testing. In particular, a $(d,e)$-disjunct graph can
uniquely distinguish between $d$-sparse vectors even if up to $\lfloor
e/2 \rfloor$ adversarial errors are allowed in the measurements.

The proof of Lemma~\ref{lem:disjunct} suggests the following decoding
procedure, that we will call the ``trivial decoder'':
\begin{quote} \textbf{Trivial Decoder:} Given a particular measurement
  outcome $y$, set the coordinate position of $x$ corresponding to
  each left vertex $i$ to be $1$ if and only if $\Gamma(i) \subseteq
  \supp(y)$.
\end{quote}
It is easy to see that this simple procedure uniquely reconstructs
every $d$-sparse vector $x$ provided that the graph is
$d$-disjunct. The trivial decoder can be adapted to the noisy case by
setting each coordinates position $i$ to be positive if and only if
$|\Gamma(i) \setminus \supp(y)| \leq \lfloor e/2 \rfloor$.

In this section, we see how graph expansion is related to
disjunctness. A left-regular bipartite graph graph with left-degree
$t$ (henceforth, \emph{$t$-regular graph}) $G(L,R,E)$ is called a \emph{$(k,
a)$-expander} if, for every left-subset $S \subseteq L$ of size at most
$k$, we have
\[
|\Gamma(S)| \geq a |S|.
\]
Obviously, we must have $a \leq t$ for this condition to be
satisfied. The parameter $a$ is called the \emph{expansion factor} and
graphs with expansion close to the degree are called \emph{lossless
  expanders}. In particular, for an \emph{error parameter} $\eps$, we
will call a $(k, t(1-\eps))$ expander graph a $(k,\eps)$-lossless
expander. The following counting argument shows that lossless
expanders are, in fact, disjunct graphs.

\begin{lemma} \label{lem:expander} Let $G=(L,R,E)$ be a $t$-regular
  $(d,\eps)$-lossless expander. Then, for every $\alpha \in [0,1)$,
  $G$ is $(d-1, \alpha t)$-disjunct provided that
  \[
  \eps < \frac{1-\alpha}{d}.
  \]
\end{lemma}

\begin{proof}
  Take any left vertex $i \in L$ and $S \subseteq L$ such that $|S|
  \leq d-1$ and $i \notin S$. By Definition~\ref{def:disjunct}, we
  need to verify that $|\Gamma(i) \setminus \Gamma(S)| > \alpha t$.
  Let $T := \Gamma(S \cup \{i\})$. Denote by $T'$ the set of vertices
  in $T$ that have more than one neighbor in $S$. By the expansion
  assumption, we know that $|T| \geq (1-\eps)dt$, implying that $|T'|
  \leq \eps dt < (1-\alpha)t$. Now we have
  \[
  |\Gamma(i) \cap \Gamma(S)| \leq |T'| < (1-\alpha)t.
  \]
  Thus,
  \[
  |\Gamma(i) \setminus \Gamma(S)| = t - |\Gamma(i) \cap \Gamma(S)| >
  \alpha t.
  \]
\end{proof}

As trivial as it is to construct non-adaptive group testing schemes
with a large number of measurements (namely, one that measures each
individual coordinate position separately), it is trivial to construct lossless
expanders (as defined above) with a large number of right vertices. In
particular, a $t$-regular $(k,0)$-lossless expander for every $k$ can
be constructed as follows: Connect each left vertex to $t$ new
vertices on the right side, so that the right degree of the graph
becomes $1$. Indeed, such a graph is $(d,t-1)$-disjunct for every $d$
and corresponds to a trivial group testing scheme. However, in order
to get any useful results, one needs to construct highly unbalanced
lossless expanders with substantially small number of right vertices.

Using the probabilistic method, Capalbo et al.~\cite{ref:CRVW02} show
that a random construction of bipartite graphs $G(L,R,E)$ with $|L| =
n$ is, with overwhelming probability, $(k,\eps)$-lossless $t$-regular
expander where $t = O((\log n)/\eps)$ and $|R| = O(kt/\eps)$. Moreover
they show that this tradeoff is about the best one can hope for.
Thus, using Lemma~\ref{lem:expander} we see that optimal expanders are
$d$-disjunct graphs with $O(d^3 \log n)$ right vertices
(measurements).  More generally, for every $\alpha \in [0,1)$ we get
$(d,e)$-disjunct matrices, where $e = \Omega(\alpha d \log
n/(1-\alpha))$, with $O(d^3 (\log n)/(1-\alpha)^2)$ right vertices.

A direct probabilistic argument, however, shows that a randomly
constructed graph (according to a carefully chosen distribution) is,
with overwhelming probability, $d$-disjunct with $O(d^2 \log(n/d))$
measurements. More generally, for every $\alpha \in [0,1)$, random
graphs are $(d,e)$-disjunct with $e = \Omega(\alpha d \log
n/(1-\alpha)^2)$ and $O(d^2 \log(n/d)/(1-\alpha)^2)$ measurements.
This tradeoff is almost optimal since known lower bounds in group
testing \cite{ref:DVMT02,ref:SW00,ref:SW04} imply that any
$(d,e)$-disjunct graph must have $\Omega(d^2 \log_d n + ed)$ right
vertices. Therefore, an optimal lossless expander achieves a number of
measurements that is off from nearly-optimal random disjunct graphs by
a factor $\Omega(d)$.

While we discussed the number of measurements achieved by random
disjunct graphs and random expanders, for applications it is generally
favorable to have a design that avoids any randomness. In particular,
a challenging goal in group testing is to come up with \emph{explicit}
constructions of measurement graphs.  The exact meaning of
``explicit'' is up to debate. One generally recognized notion of
explicitness is the existence of a deterministic algorithm that
outputs the adjacency matrix of the measurement graph in polynomial
time with respect to its size.  A more stringent requirement would be
to have a deterministic algorithm that, given integer parameters $i,
j$, outputs the index of the $j$th neighbor of the $i$th left vertex
of the graph in polynomial time in the bit representation of $(i,j)$
(i.e., $\poly(\log n)$ where $n$ is the dimension of the sparse vector
to be measured).

The state-of-the-art explicit constructions of lossless expanders
still do not attain the optimal parameters.  For our applications,
some notable explicit lossless expanders include:

\begin{itemize}
\item Zig-Zag based $(k,\eps)$-lossless expanders due to Capalbo et
  al.~\cite{ref:CRVW02}: Achieves degree
  $t=2^{O(\log^3(\log(n)/\eps))}$ and $|R| = O(kt/\eps)$ right
  vertices.

\item Coding-based expander of Guruswami et al.~\cite{ref:GUV09}: For
  every constant parameter $\gamma > 0$, achieves degree \[
  t=O\left(((\log n)(\log k)/\eps)^{1+1/\gamma}\right) \] and right
  part size $|R| \leq t^2 k^{1+\gamma}$.
\end{itemize}

Using Lemma~\ref{lem:expander}, the two constructions result in
explicit $d$-disjunct graphs with respectively $d^2 \qpoly(d \log n)$
and $O(d^4 \log^2 n \log^2 d)$ measurements (by setting $\gamma:=1$).
Analogous expressions can be obtained for the noise-tolerant case as
well.  We remark that explicit nearly optimal disjunct graphs (with
$O(d^2 \log n)$ measurements) can be obtained from the recent
construction of Porat and Rothschild \cite{ref:PR08}. This
construction is however not explicit in the more stringent sense
discussed above.  The classical work of Kautz and Singleton
\cite{ref:KS64} (that uses Reed-Solomon codes as the main ingredient)
can be used to construct fully explicit $d$-disjunct graphs with
$O(d^2 \log^2 n)$ measurements, which is fairly sub-optimal.

\section{Graphs from Condensers and Extractors}
\label{sec:noise}
\newcommand{\bgam}{\boldsymbol{\Gamma}}

A nice way of thinking about a $(k,\eps)$-lossless expander is through
\emph{injectivity}: for every subset $S$ of left vertices, where $|S|
\leq k$, the neighborhood $\Gamma(S)$ has little \emph{collisions}.  Namely,
almost all vertices in $\Gamma(S)$ are connected to only one vertex in
$S$. Therefore, if the $j$th neighbor of a vertex $v \in S$ is
connected to $v' \in \Gamma(S)$, from $v'$ one can almost always
uniquely recover $v$ and $j$.

An injective map preserves \emph{entropy}. Thus, the above discussion
can be rephrased in information-theoretic terms.  Denote by $\bgam(S)$
the probability distribution induced on the set of right vertices by
picking a uniformly random neighbor of a uniformly random left vertex
in $S$.  For a $t$-regular graph, the entropy of the distribution
induced on the \emph{edges} of the graph by the above sampling
procedure is $\log (|S| t)$. The almost-injectivity property of the
graph intuitively implies that this entropy must be almost preserved
in $\bgam(S)$. In fact, the intuition can be made precise to show that
$\bgam(S)$ is $\eps$-close to a distribution with entropy $k$
\cite{ref:TUZ01}.  Here, the measure of distance is \emph{statistical
  distance}: Two distributions are $\eps$-close if and only if the
probability that they assign to any event is different by at most
$\eps$.  The measure of entropy is the notion of \emph{min-entropy}
which lower bounds Shannon entropy: A distribution on a finite domain
has min-entropy $\log k$ if and only if the probability that it
assigns to each element of the sample space is upper bonded by $1/k$.

The information theoretic interpretation of lossless expanders suggest
the following generalized notion: Call a $t$-regular graph $G(L,R,E)$
a $k \to_\eps k'$ condenser if, for every $S \subseteq L$ of size $k$,
the distribution $\bgam(S)$ induced on $R$ is $\eps$-close to a
distribution with entropy $\log (tk')$. Therefore, a
$(k,\eps)$-lossless expander is a $\underline{k} \to_\eps
\underline{k}$ condenser for every $\underline{k} \leq k$.

A particularly interesting special case is when $k' = |R|$.  In this
case, the output distribution $\bgam(S)$ becomes almost uniform on the
set of right vertices. A $k \to_\eps |R|$ condenser is called a
\emph{$(k, \eps)$-extractor}.

In the previous section, we saw that lossless expanders are disjunct
graphs as long as the error is sufficiently small; namely, smaller
than about $1/d$. If we allow a larger, and in particular, constant
error, we cannot hope for a disjunct graph since the number of right
vertices would be allowed to violate the known lower bounds for
disjunct graphs.  However, in this section we see that such graph are
still able to \emph{approximate} sparse vectors, even in highly noisy
settings. The key idea is captured by the following lemma.

\begin{lemma} \label{lem:extractor} Let $G(L,R,E)$ be a $t$-regular
  $(k,\eps)$-extractor.  Then, for every $d$-sparse vector $x \in
  \zo^n$ where $n := |L|$ the following holds provided that $dt <
  |R|(1-\eps)$: Given $y := G[x]$, the trivial decoder outputs $x' \in
  \zo^n$ such that $\supp(x) \subseteq \supp(x')$ and $|\supp(x')| <
  k$.
\end{lemma}

\begin{proof}
  By the way the trivial decoder is designed, it obviously does not
  output any false negatives; i.e., we are ensured to have $\supp(x)
  \subseteq \supp(x')$.  Let $S := \supp(x')$ and suppose now, for the
  sake of contradiction, that $|S| \geq k$. Thus the extractor
  property ensures that the distribution $\bgam(S)$ is $\eps$-close to
  the uniform distribution on $R$.

  Now consider the event $T := \supp(y) \subseteq R$.  The probability
  mass assigned to this event by the distribution $\bgam(S)$ is equal
  to $1$ since the trivial decoder is defined so that for each $i \in
  \supp(x')$, we have $\Gamma(i) \subseteq T$.  On the other hand, the
  probability assigned to $T$ by the uniform distribution on $R$ is
  $|\supp(y)|/|R|$, which is at most
  \[
  \frac{t \cdot |\supp(x)|}{R} \leq \frac{td}{R}.
  \]
  Since $\bgam(S)$ is $\eps$-close to uniform, it must be that
  \[
  \frac{td}{R} \geq 1-\eps,
  \]
  contradicting the assumption.
\end{proof}

The above lemma can be extended to arbitrary $k \to_\eps k'$
condensers, in which case the required tradeoff would become $d <
k'(1-\eps)$ (through a similar line of argument).  Since, obviously,
for any condenser one must have $k \geq k'$, the bound $k-d-1$ on the
number of false positives in the approximation output by the trivial
decoder can be minimized by taking the measurement graph to be a
lossless expander (so that $k = k'$). In particular, by letting $\eps
\ll 1/(d+1)$ one can recover the statement of
Lemma~\ref{lem:expander}.  However, a constant $\eps$ (even, say,
$\eps=1/2$) may still keep the amount of false positives in the
reconstruction bounded by $O(d)$.

Same as lossless expanders, the probabilistic method can be used to
show that $(k,\eps)$-extractor graphs exist with degree $t = O(\log
(n-k)/\eps^2)$ and $|R| = \Omega(\eps^2 tk)$ right vertices. Moreover,
this is about the best tradeoff to hope for \cite{ref:lowerbounds}.

Using an optimal extractor or an optimal lossless expander in the
result discussed above gives measurement graphs with $O(d \log
n)$ right vertices for which the trivial decoder results in only
$O(d)$ false positives in the reconstruction.  At the cost of a loss
in the constant factors, the amount of false positives can be kept
bounded by $\delta d$ for any arbitrary constant $\delta > 0$ when a
lossless expander is used.

A non-adaptive scheme as above can be used in a so-called
\emph{trivial two-stage schemes} \cite{ref:Kni95} as follows: After
obtaining a set of size $O(d)$ of \emph{candidate positives}, one can
apply individual tests on the elements of this set to identify the
exact set of positives. For most practical applications, such
schemes are as good as fully non-adaptive schemes.

The result given by Lemma~\ref{lem:extractor} is extended in
\cite{ref:Che09} to not only general condensers, but also highly noisy
settings when both false positives and false negatives may occur in
the measurement outcomes. When false negatives are allowed in the
measurements, the trivial decoder should be slightly altered to
include those coordinate positions in the support of the
reconstruction that have a sufficient ``agreement'' with the
measurement outcomes. We omit the details in this exposition, but
instead state the tradeoffs obtained when the result is instantiated
with optimal extractors and lossless expanders:

\begin{itemize}
\item An optimal extractor can be set to tolerate \emph{any} constant
  fraction $p \in [0,1)$ of false positives in the measurements (i.e.,
  when up to $p$ fraction of the measurement outcomes may
  adversarially be flipped from $0$ to $1$) and an $\Omega(1/d)$
  fraction of false negatives.  The reconstruction is guaranteed to
  output a sparse vector containing the support of the original vector
  $x$ and possibly up to $O(d)$ additional \emph{false positives}.

\item An optimal lossless expander can be set to tolerate some
  constant fraction of false positives and $\Omega(1/d)$ fraction of
  false negatives in the measurement outcomes and still reconstruct
  any $d$-sparse vector up to $\delta d$ false positives, for any
  arbitrarily chosen constant $\delta > 0$.
\end{itemize}

We see that, while optimal extractors offer a better noise resilience
compared to optimal lossless expanders, the latter is more favorable
when a fine approximation of the unknown sparse vector is sought for.
The above-mentioned parameters achieved by optimal extractors and
lossless expanders are essentially optimal \cite{ref:Che09}.

Same as lossless expanders, known explicit extractors still do not
match non-constructive parameters. Notable explicit extractors for our
applications include:

\begin{itemize}
\item Coding-based extractor of Guruswami et al.~\cite{ref:GUV09}:
  Achieves degree
  \begin{eqnarray*}
    t &=& O((\log n) \cdot 2^{O(\log \kappa \cdot \log(\kappa/\eps))})\\
    &=& O((\log n) \cdot \qpoly(\log k)),
  \end{eqnarray*}
  where $\kappa := \log k$, and right size
  \[
  |R| = \Omega(\eps^2 tk).
  \]

\item Trevisan's extractor \cite{ref:Tre,ref:RRV}: Achieves
  \[t=2^{O(\log^2(\log(n)/\eps) \cdot \log \kappa)}\] and $|R| =
  \Omega(\eps^2 tk)$ right vertices.
\end{itemize}

The trade-offs obtained by various choices of the underlying extractor
and expander is summarized in Table~\ref{tab:resilient}.  While the
parameters obtained by the graphs based on Trevisan's extractor and
the lossless expander of Guruswami et al.~are superseded by other
constructions, it can be shown that \cite{ref:Che09} these graphs
allow a more efficient reconstruction algorithm than the trivial
decoder; namely, one that runs in polynomial time with respect to the
number of measurements (a quantity that can be in general
substantially lower than the running time $O(|L|\cdot|R|)$ of the
trivial decoder).

\begin{table} [!t] %[p]
  \caption[Summary of the noise-resilient group testing schemes]{A summary of constructions in Section~\ref{sec:noise}. The parameters $\alpha \in [0,1)$ and
    $\delta > 0$ are arbitrary constants, $m$ is the number of measurements,
    $e_0$ (resp., $e_1$) the number of tolerable false positives (resp., negatives) in the measurements,
    and $e'_0$ is the number of false positives in the reconstruction. 
    The underlying condenser corresponding to each row is:
    (1)~optimal extractor, (2)~optimal lossless expander, (3)~extractor of 
    Guruswami et al.~\cite{ref:GUV09}, (4)~lossless expander of Capalbo et al.~\cite{ref:CRVW02},
    (5)~Trevisan's extractor~\cite{ref:Tre,ref:RRV}, (6)~lossless expander of
    Guruswami et al.~\cite{ref:GUV09}.
  }
  \begin{center}
    \begin{tabular}{|l|c|c|c|c|}
      \hline
      & $m$ & $e_0$ & $e_1$ & $e'_0$ \\
      \hline
      1&$O(d \log n)$ & $\alpha m$ & $\Omega(m/d)$ & $O(d)$  \\
      2&$O(d \log n)$ & $\Omega(m)$ & $\Omega(m/d)$ & $\delta d$ \\
      3&$O(d^{1+o(1)} \log n)$ & $\alpha m$ & $\Omega(m/d)$ & $O(d)$ \\
      4&$d \cdot \qpoly(\log n)$ & $\Omega(m)$ & $\Omega(m/d)$ & $\delta d$ \\
      5&$d \cdot \qpoly(\log n)$ & $\alpha m$ & $\Omega(m/d)$ & $O(d)$ \\
      6&$\poly(d) \poly(\log n)$ & $\poly(d) \poly(\log n)$ & $\Omega(e_0/d)$ & $\delta d$ \\
      \hline
    \end{tabular}
  \end{center}
  \label{tab:resilient}
\end{table}

\section{Expanders and the Threshold Model}
\label{sec:threshold}

Lossless expanders are used in \cite{ref:Che10} to construct
measurement graphs suitable for the threshold model, where the
threshold $u$ can be an arbitrary constant. This result applies to the
positive-gap case as well as the gap-free case. However, in this
section we focus our attention to the gap-free model.

The combinatorial property needed for the measurement graphs suitable
for the $u$-threshold model is an extension of disjunctness, defined
below.

\begin{defn} 
  A bipartite graph $G(L,R,E)$ is called \emph{$(d,e;u)$-disjunct} (or
  \emph{threshold-disjunct}) if, for every left vertex $i \in L$
  (called the \emph{special vertex}), every set $S \subseteq L$
  containing $i$ (called the \emph{critical set}) such that $u \leq
  |S| \leq d$, and every $Z \subseteq L$ disjoint from $S$ (called the
  \emph{zero set}) such that $|Z| \leq |S|$, we have
  \begin{equation} \label{eqn:udisjunct} |\{ v\in \Gamma(i)\colon
    |\Gamma(v) \cap Z| = 0 \land |\Gamma(v) \cap S| = u \}| > e.
  \end{equation}
  In other words, more than $e$ neighbors of $i$ must have no
  neighbors in $Z$ and exactly $u$ neighbors (including $i$) in $S$.
  The parameter $e$ is called the \emph{noise tolerance}.
  \label{def:udisjunct}
\end{defn}

A simple combinatorial trick allows us to reduce the problem of
designing group testing schemes for the threshold model (i.e.,
construction of $(d,e;u)$-disjunct graphs) to the same problem in
classical group testing (i.e., $(d,e)$-disjunct graphs). This is done
through a direct product defined below.

\newcommand{\rep}{\odot}
\begin{defn}
  Let $G_1(L,R_1,E_1)$ and $G_2(L,R_2,E_2)$ be graphs with the same
  set of left vertices. Then the product $G_1 \rep G_2$ is a graph
  $G_3(L, R_3, E_3)$ with $R_3 := R_1 \times R_2$ in which a vertex
  $(i,j) \in R_3$ is connected to $v \in L$ if and only if either $i
  \in R_1$ in $G_1$ or $j \in R_2$ in $G_2$ is connected to $v$.
\end{defn}

Disjunct graphs for the $u$-threshold model can be constructed by
taking the product of ordinary disjunct graphs with graphs satisfying
a certain combinatorial property, that here we call \emph{regularity}.
The exact definition of regular graphs is very similar to that of
threshold-disjunct graphs. Formally, a \emph{$(d,e;u)$-regular} graph
is defined exactly as in Definition~\ref{def:udisjunct}, except that
the requirement \eqref{eqn:udisjunct} is modified to
\begin{equation} \label{eqn:uregular} |\{ v\in R\colon |\Gamma(v) \cap
  Z| = 0 \land |\Gamma(v) \cap S| = u \}| > e.
\end{equation}
That is, there is no ``special'' vertex $i$ this time and the only
requirement from the graph is that, for every choice of a critical set
$S$ and a zero set $Z$ as in Definition~\ref{def:udisjunct}, there
must be more than $e$ right vertices that are each connected to
exactly $u$ vertices in $S$ and none of the vertices in $Z$.

The following is proved in \cite{ref:Che10}:

\begin{lemma} \label{lem:rep} Let $G_1$ and $G_2$ be bipartite graphs
  with $n$ left vertices, such that $G_1$ is
  $(d-1,e_1;u-1)$-regular. Let $G := G_1 \rep G_2$, and suppose that
  for $d$-sparse boolean vectors $x, x' \in \zo^n$ such that $\wgt(x)
  \geq \wgt(x')$, we have
  \[
  | \supp(G_2[x]_1) \setminus \supp(G_2[x']_1)| \geq e_2.
  \]
  Then, $ |\supp(G[x]_u) \setminus \supp(G[x']_u)| \geq (e_1+1) e_2.
  $ \qed
\end{lemma}

Thus, if a measurement graph is able to distinguish between $d$-sparse
vectors in the classical model of group testing (i.e., with threshold
$1$), then its product with a $(d-1,e;u-1)$-regular matrix
distinguishes between $d$-sparse vectors in the $u$-threshold
model. In fact it turns out that if the original graph is disjunct in
the classical sense, the product becomes threshold-disjunct for
threshold $u$. Thus in order to design measurement schemes for the
threshold model, it suffices to focus on construction of regular
graphs.

A construction of regular graphs based on lossless expanders is given
in \cite{ref:Che10}. The lossless expanders required by this
construction have to satisfy a certain property, and we use the term
\emph{function graph} to refer to such graphs. A function graph
$G(L,R,E)$ is a $t$-regular bipartite graph where the set $R$ of the
right vertices is partitioned into $t$ equal-sized groups. The
requirement is that the $t$ neighbors of each left vertex must each
belong to a distinct group. All the above-mentioned explicit, and
probabilistic, constructions of lossless expanders are in fact
function graphs.

The construction can be conveniently explained using the following
graph composition: Let $G_1(L,[t] \times R_1,E_1)$ be a $t$-regular
function graph where the right nodes are partitioned into $t$ groups
of size $|R_1|$ each, and $G_2(R_1, R_2, E_2)$ be a bipartite graph.
Then the composition $G_1 \leadsto G_2$ is a bipartite graph $G_3(L,
[t] \times R_2, E_3)$ such that, for each $i \in L, j \in [t], k \in
R_2$, an edge $(i,(j,k))$ is in $E_3$ if and only if there is a $v \in
R_1$ such that $(i,(j,v)) \in R_1$ and $(v,k) \in R_2$.  Intuitively,
the composition can be seen as follows: Each of the $t$ groups of the
right vertices in $G_1$ is replaced by a copy of $G_2$, so that a
two-layered graph is obtained. Then the two layers are collapsed into
one by short-cutting all paths of length two from left to right.

\begin{figure} %(intermediate construction) %\hfill
  \caption{Construction of regular matrices.}

  \begin{framed}
    \begin{itemize}
    \item {\it Given: } A $t$-regular $(k, \eps)$-lossless expander
      $G(L,R,E)$ where $k, |L|, |R|$ are powers of two, integer
      parameter $u \geq 1$ and real parameter $p \in [0,1)$ such that
      $\eps < (1-p)/16$,

    \item {\it Output: } A measurement graph with $n := |L|$ left
      vertices and $ m := O_u(tk (|R|/k)^u) $ right vertices.

    \item {\it Construction: }
      % Let $r := \lceil \log(d/u) \rceil$. We construct $r$ matrices
      % $\cM_1, \ldots, \cM_r$ that are stacked on top of one another
      % to
      % form the resulting matrix $\cM$. For each $i \in [r]$, we
      % construct $\cM_i$ as follows:
% 
      Let $G_1=(R, R_1, E_1)$ be any bipartite bi-regular graph with
      $|R_1|=k$, left degree $d_\ell := 8u$, and right degree $d_r :=
      8u(|R|/k)$.  Replace each right vertex $v$ of $G_1$ with
      $\binom{d_r}{u}$ vertices, one for each subset of size $u$ of
      the vertices on the neighborhood of $v$, and connect them to the
      corresponding subsets. Denote the resulting graph by $G_2 = (R,
      R_2, E_2)$, where $|R_2| = k \binom{d_r}{u}$. Output $G \leadsto
      G_2$.
    \end{itemize}
  \end{framed}
  \label{fig:regular}
\end{figure}

Using the above notation, a construction of regular matrices is
described in Fig.~\ref{fig:regular}.  Analysis of the construction
leads to the following result that is proved in \cite{ref:Che10}:

\begin{theorem} \label{thm:regular} The graph output by the
  construction described in Fig.~\ref{fig:regular} is $(k/2, pt;
  u)$-regular as long as, in the definition of regularity, the
  critical set $S$ is restricted to have size at least $k/4$. \qed
\end{theorem}

In order to obtain $(d,pt;u)$-regular graphs, by
Theorem~\ref{thm:regular} it suffices to apply the construction of
Fig.~\ref{fig:regular} for $O(\log d)$ different values of $k$,
namely\footnote{ The case where the sparsity (the size of the critical
  set) lies between $u$ and $k_0-1$ is of minor importance and can be
  handled using straightforward tricks.  },
\[ k = 2^{\lceil \log d \rceil+1}, 2^{\lceil \log d \rceil}, 2^{\lceil
  \log d \rceil -1}, \ldots, 2^{\lceil \log u \rceil+2} =: k_0.\] By
doing so, we obtain $O(\log d)$ graphs. We take the \emph{union} of
all the obtained graphs (where for two graphs $G_1(L,R_1,E_1)$ and
$G_2(L,R_2,E_2)$ with the same set of left vertices, the union is a
graph $G_3(L,R_1 \cup R_2,E_3)$ where $e \in E_3$ if and only if $e
\in E_1$ or $e \in E_2$).  The resulting graph must be
$(d,pt;u)$-regular, since for every possible size of the critical set,
Theorem~\ref{thm:regular} applies for at least one of the components
of the union.

\begin{table*}
  \caption{Summary of the parameters achieved by the threshold testing scheme described in Section~\ref{sec:threshold},
    for various choices of the lossless expander. The 
    noise parameter $p \in [0,1)$ is arbitrary. The underlying expander corresponding to each row is:
    (1)~an optimal lossless expander, (2)~construction of Capalbo et al.~\cite{ref:CRVW02},
    (6)~construction of Guruswami et al.~\cite{ref:GUV09}.
  }

  \begin{center}
    \begin{tabular}{|c|l|l|p{10cm}|}
      \hline
      & Number of rows & Tolerable errors & Remarks \\ \hline
      1 & $O(d^{3} \frac{(\log d) \log^2 n}{(1-p)^{2}})$ & $\Omega(p d^2 \frac{\log^2 n}{(1-p)^2})$& \\
      2 & $O(d^{3} \frac{(\log d) T_2 \log n}{(1-p)^{g+2}})$ & $\Omega(p d^2 \frac{T_2 \log n}{1-p})$ & 
      $T_2 = \exp(O(\log^3 \log n )) = \qpoly(\log n)$. \\
      3 & 
      % $\begin{array}{l} O(d^{g+3+\beta} \cdot \\ \quad\qquad \cdot
      %   \frac{T_3^{\ell} \log n}{(1-p)^{g+2}}) \end{array}$
      $O(d^{3+\beta} \frac{T_3^{\ell} \log n}{(1-p)^{g+2}})$
      % & $\begin{array}{l} \Omega(p d^{2-\beta} \frac{ \log n}{1-p}
      %   ) \end{array}$ &
      & $\Omega(p d^{2-\beta} \frac{ \log n}{1-p} )$ & $\beta>0$ is
      any arbitrary constant and
      $T_3 = ((\log n)(\log d))^{1+u/\beta} = \poly(\log n, \log d)$.  \\
      \hline
    \end{tabular}
  \end{center}
  \label{tab:threshold}
\end{table*}

Using optimal lossless expanders, the construction leads to $(d, e;
u)$-regular graphs with $O(d(\log d)(\log n))$ right vertices and $e =
\Omega(d \log n)$.  By taking the direct product of such graphs with
optimal explicit $(d, \Omega(d \log n))$-disjunct graphs of Porat and
Rothschild \cite{ref:PR08} that have $\Omega(d^2 \log n)$ right
vertices, we would get $(d,\Omega(d^2 \log^2 n);u)$-disjunct graphs
with $O(d^3 (\log d) (\log^2 n))$ right vertices. The amount of
measurements achieved by different choices of the underlying condenser
is summarized in Table~\ref{tab:threshold}.

A direct probabilistic argument can be used to show that, random
regular graphs (sampled according a suitably chosen distribution) are,
with overwhelming probability, $(d,e;u)$-disjunct with $m=O(d^2 (\log
d) \log(n/d) / (1-p)^2)$ and $e=\Omega(p d
\log(n/d)/(1-p)^2)$. Moreover, since threshold disjunct graphs are
stronger objects than classical disjunct graphs, the above-mentioned
lower bound of $m=\Omega(d^2 \log_d n+e d)$ applies to them as well,
implying that the probabilistic construction is nearly optimal.  Thus,
the construction described in this section is off by a factor
$\Omega(d)$ in the number of measurements even when an optimal
lossless expander is available. Closing this gap is an interesting
question for future research.

%%%%%%%%%%%%%%%%%%%%%%%%%%%%%%%%%%%%%%%%%%%%%%%%%%%%%%%%%%%%
\addtolength{\textheight}{-11cm} % This command serves to balance the column lengths
% \addtolength{\textheight}{-5.0cm} % This command serves to balance the column lengths

% \bibliographystyle{plain} \bibliography{bibliography}
\bibliographystyle{IEEEtran} \bibliography{bibliography}

\end{document}